\documentclass[journal]{IEEEtran}
\usepackage{amsthm}

\usepackage{cite}
\usepackage{amsmath,amssymb,amsfonts}
\usepackage{algorithm}
\usepackage{algpseudocode}
\usepackage{comment}
\newcommand{\ket}[1]{\left| #1 \right\rangle}
\usepackage{graphicx}
\usepackage{textcomp}
\usepackage{braket}
\usepackage{xcolor}
\usepackage[T1]{fontenc}
\usepackage{lmodern}

\usepackage{listings}
\usepackage{xcolor}
\usepackage{caption} 
\usepackage{subcaption}
\usepackage{tabularx}
\usepackage{array}
\usepackage{multirow}
\usepackage{amsthm}

\newtheorem{theorem}{Theorem}
\newtheorem{lemma}{Lemma}

\newcolumntype{C}[1]{>
{\centering\arraybackslash}m{#1}}

\def\BibTeX{{\rm B\kern-.05em{\sc i\kern-.025em b}\kern-.08em
    T\kern-.1667em\lower.7ex\hbox{E}\kern-.125emX}}
\begin{document}

\title{On Reachability Problem in Bounded Quantum Petri Nets}

\author{ {\bf Syed Asad Shah } \\ Department of Computer Engineering \\ Bilkent University, Ankara, Turkey \\
\and {\bf   A. Yavuz Oru\c{c}} \\ Department of Electrical and Computer Engineering \\ University of Maryland, College Park, MD, USA}

\maketitle

\begin{abstract}
In this paper, we propose a novel quantum solution to address the problem of reachability in bounded quantum Petri nets (QPNs)---an advanced modeling framework that combines the classical Petri net model with quantum mechanical principles. The proposed approach exploits quantum parallelism to construct a superposition over all reachable markings from the initial marking. Grover’s amplitude amplification algorithm is then applied to efficiently identify a desired target marking, while ancillary q-tokens used solely for transition control are excluded from the search space, significantly reducing its size. Theoretical analysis shows that our approach achieves a quadratic speed-up over classical exhaustive algorithms. Experimental results also  confirm the correctness and feasibility of the proposed quantum algorithm for solving the bounded reachability problem in Quantum Petri Nets.
\end{abstract}


\section{Introduction}

\noindent Over the past forty years, quantum computing has developed into an advanced and rapidly evolving discipline that integrates concepts from computer science, physics, and mathematics \cite{Lanzagorta2022, Djordjevic2021, Hassija2020, Horowitz-2019, Wolfgang2019}. Quantum computing is appealing because quantum mechanics can provide an almost unlimited level of parallelism at small physical scales, a capability that classical computers cannot achieve without replicating  hardware resources. Not all complex computational tasks may benefit from quantum parallelism but it has been established that for  a sizable subset of problems such as factoring large integers, carrying out constrained and unconstrained search, and completing optimization, quantum computing provides an efficient method of solution to circumvent the exponential resource growth required by classical methods \cite{Shor-1994, Grover-1996, Rhonda-2023}.

Several areas within quantum computing are currently the focus of active research \cite{Preskill-2023, Ramezani-2020, Aumasson-2017, Humble-2021, Bacon-2010}. The main goal of our study in this paper is to address the problem of reachability within the context of quantum Petri nets. A Quantum Petri Net (QPN) model is built on the foundation of classical Petri net model, which is a well-known framework to model concurrent, distributed, and event-driven systems  \cite{Papavarnavas, Proth}. Classical Petri nets have been widely used in domains such as communication networks, workflow systems, and manufacturing processes, because they effectively model synchronization, concurrency, and resource sharing problems \cite{Murata-1989, Zurawski-1994, Letia-2022}. However, such Petri nets are inherently limited to modeling deterministic or stochastic classical processes and thus cannot naturally capture the distinctive behaviors of quantum systems. By integrating principles of quantum mechanics into the classical Petri net token model, QPN permits modeling quantum states, such as superposition and entanglement, together with their stochastic evolution in time. This extended framework enables the representation of complex quantum behaviors that cannot be effectively captured by classical Petri nets. 

To this end, we consider a central analytical problem in classical Petri nets, namely the  reachability in such nets, i.e., to determine whether a desired marking (state) can be obtained from an initial marking through a finite sequence of transitions \cite{Ye-2003, Ru-2009, Kostin-2003}. The extreme complexity of the unbounded token case is well-known and we refer the reader to  \cite{Lipton-1976,Mayr-1981,Czerwinski-2020,Wang} for the study and mitigation of the state-space explosion problem in unbounded Petri nets.   In this paper, we focus on conservative Petri nets in which the number of tokens is fixed. We introduce a quantum-inspired computational approach that leverages quantum parallelism to explore reachable markings efficiently in such Petri nets using the QPN  model.  For a Petri net with $n$ tokens and $m$ places, a simple computation shows that the number of Petri nets is bounded by $\binom{n+m-1}{m-1}$ and for a Petri net with $k$ transitions, it may take \(O(\binom{n+m-1}{m-1} k m )\) steps to determine if a target marking is reachable from a source marking even with the use of DFS and BFS  algorithms. Our theoretical analysis establishes  that quantum parallelism provides a quadratic speed-up over such algorithms, demonstrating the utility and benefit of the quantum Petri net model. 

The remainder of the paper is organized as follows. The next section describes the QPN model. Section III formalizes the quantum reachability problem and provides our quantum reachability search algorithm. Section IV provides a case study to demonstrate the use of our algorithm and the paper is concluded in Section V.

\vspace{-10pt}
\section{Quantum Petri Net Model}

\noindent The Quantum Petri Net (QPN) model extends the classical Petri net model by integrating quantum mechanical principles including quantum state representation, quantum dynamics, and intrinsic probabilistic behavior into the classical Petri net model. The QPN model used here was derived from our earlier work\cite{Shah-2025}.

\noindent
\textbf{Definition 1}: A Quantum Petri net $(QPN)$ is a \text{6-tuple} $(D, P, T, E,  \mu_t,  v_t)$ where: 

\begin{enumerate}
\item $D$ is a finite set of quantum tokens $\{d_1, d_2, ..., d_n\}, n\!\!\geq\!\! 1,$  henceforth to be referred to as $q$-tokens, 
\item $P$ is a finite set of places $\{P_1, P_2, P_3, \ldots, P_m\}$, $ m \geq 1, $
\item $T$ is a finite set of transitions $\{T_1, T_2, T_3, \ldots, T_k\}$, $ k \geq 1, $ 
\item $E$ is a finite set of directed and labeled arcs, $\{x_1,x_2,\ldots,x_z\}, z\ge 1$ that connect places in $P$ with transitions in $T$  together with  the number of $q$-tokens needed for a transition to fire or to consume,
\item $\mu_t,$   called a {\em marking}, is a mapping from $P$ to the power set of $D$, where $\mu_t(P_i)$ denotes the set of  $q$-tokens residing in place $P_i$ at time $t=0,1,2,\ldots$. 
%and address field is the order in which transitions are fired to reach that specific marking i,
\item \textit{$v_t(P_i)$} is an assignment of qubits to $q$-tokens in place $P_i$  in marking $\mu_t(P_i)$ at time $t = 0, 1, 2, ... .$
\end{enumerate}

\noindent Quantum tokens or $q$-tokens represent the qubits or multi-qubits in the QPN system. It's assumed that tokens are conserved as the QPN evolves over time. These $q$-tokens are visualized as small black filled circles stored within the larger empty circles called places. Places can store a certain number of $q$-tokens using a primitive building block called a quantum S-R (Q-S-R) flip-flop presented in our earlier work \cite{Shah-2025}. The Q-S-R flip-flop stores a qubit rather than a bit as in a classical S-R flip-flop. Transitions are visualized as a large rectangle, perform quantum operations to modify the state of $q$-tokens and position of $q$-tokens using quantum gates. Directed arcs that connect places to transitions are referred to as incoming arcs, and those that connect transitions to places are referred to as outgoing arcs. These arcs facilitate the flow of $q$-tokens between places, and the labels on them are placeholders for $q$-tokens that are fed into a transition or a place with indicated quantities. The distribution of $q$-tokens across the places and the order in which transitions are fired represent the various states of the QPN system, referred to as marking. The function $v(P_i,\mu(t))$ represents the assignment of qubits to $q$-tokens under the marking $\mu$ at time $t$ in place $P_i$. With each transition execution, the QPN system evolves; this evolution is represented by a  {\em firing sequence}, i.e., a set of valid transitions of the QPN.

Figure.~\ref{fig:spn} depicts the simple QPN model with $n = 4$ $q$-tokens  $d_1, d_2, d_3, d_4$,  $m = 4$ places  $P_1, P_2, P_3, P_4$, $k = 2$ transitions  $T_1, T_2$, and $5$ directed arcs $x_1, x_2, x_3, x_4, x_5,$ with their capacities indicated next to each. An example of the initial marking of the QPN, where the four tokens are initialized to $d_1 = |00\rangle, d_2 = |10\rangle, d_3 = |00\rangle, d_4 = |01\rangle,$ at time $t = 0$ is shown below:

\begin{figure}[]
\centering
\includegraphics[height=6 cm, width=7 cm]{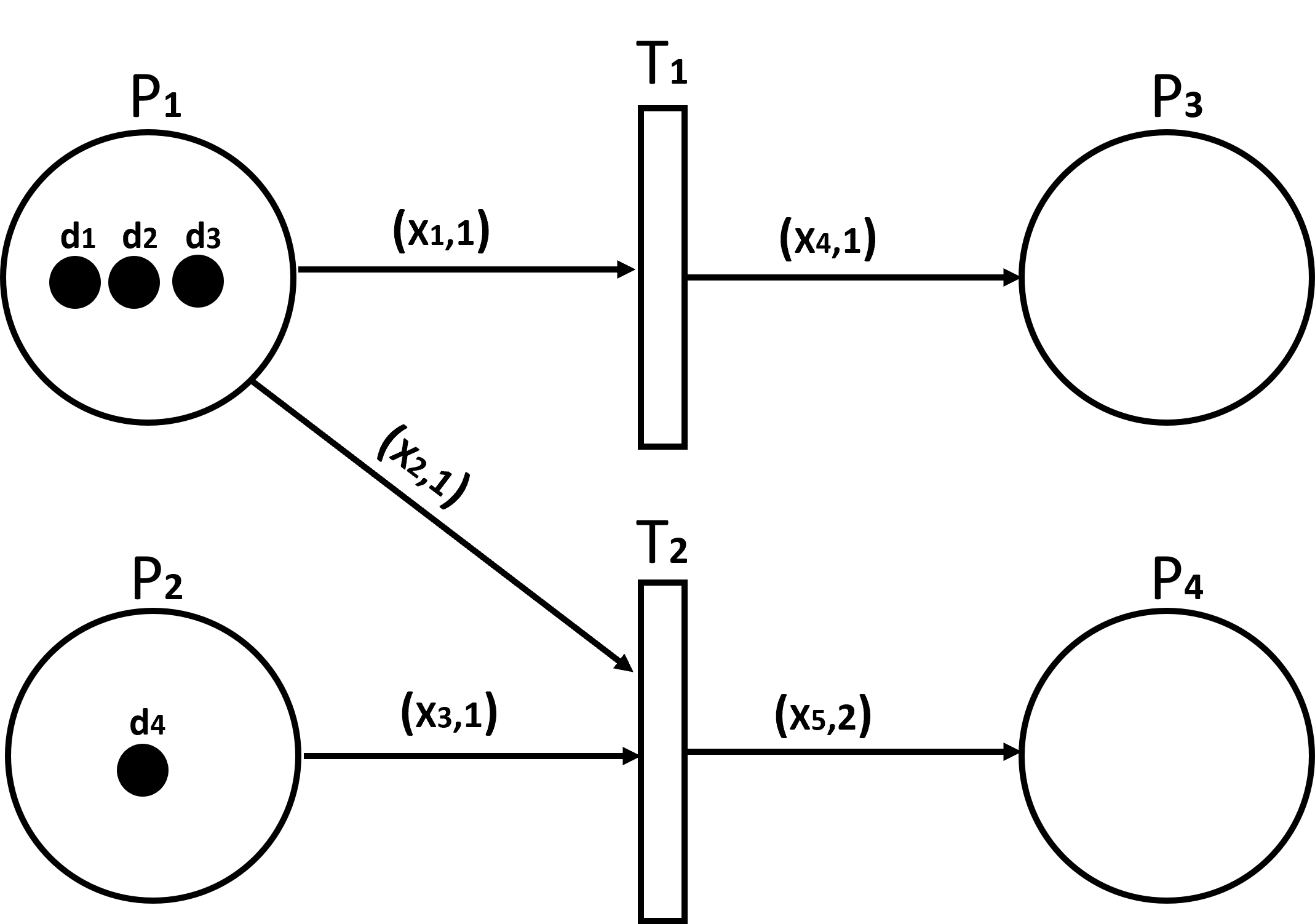}
\caption{A simple QPN example with $m = 4$ places, $k = 2$ transitions and $n = 4$ $q$-tokens.}
\vspace{-30pt}
\label{fig:spn}
\end{figure} 

\begin{align*}
\mu_0(P_1) &= \{d_1,d_2,d_3\}, &
\mu_0(P_2) &= \{d_4\}, \\
\mu_0(P_3) &= \emptyset, &
\mu_0(P_4) &= \emptyset.
\end{align*}

\vspace{-10pt}
\begin{align*}
v_0(P_1) &= 
\left\{
\begin{array}{l}
d_1: |00\!>  \\ 
d_2: |10\!> \\ 
d_3: |00\!> \\
\end{array} 
\right\}, \\
v_0(P_2) &= 
\,\,\left\{
\begin{array}{l}
d_4: |01\!\!> 
\end{array}
\right\}.
\end{align*}

\vspace{-10pt}
\section{The Quantum Reachability Problem}
\noindent This study aims to address the reachability problem in the context of quantum Petri nets, which is defined as follows. \\
\textbf{Definition 2}:
 \noindent
Given a Quantum Petri Net
\[
\mathcal{QPN} = (D, P, T, E, \mu_t, v_t),
\]
the reachability problem asks whether a desired marking $\mu_d$ is reachable from the
initial marking $\mu_0$. Formally, $\mu_d$ is reachable from $\mu_0$ if there exists a finite
firing sequence $\sigma$ such that 
\[
\mu_0 \xrightarrow{\;\sigma\;} \mu_d .
\]

Before we proceed with the description of the proposed algorithm, we first introduce the following definition.

\noindent\textbf{Definition 3:}
Let $L=\binom{n+m-1}{m-1}.$ It is easily established that  \(L\) is the maximum number of possible markings  with \(n\) tokens and  \(m\) places and sets an upper bound on the number of iterations in Phase 1 of Algorithm 1.

\vspace{5pt}\noindent
{\bf Algorithm 1. Quantum Reachability}

\noindent{\bf Input:} {\rm A bounded quantum Petri net
with initial marking $\mu_0$, and target marking $\mu_d$.}\\
{\bf Output:} {\rm Report whether $\mu_d$ is reachable from $\mu_0$ with high probability.}\\

\vspace{-5pt}\noindent{\bf Phase I: Superposition over Reachable Markings}\\
\indent
{\bf Step 1:}
{\rm Define $p$ ancillary transition registers
$\alpha_1,\alpha_2,\ldots,\alpha_p$, where $p\le L$ and  each register consists of
$\lceil\log_2(|T|+1)\rceil$ qubits\footnote{ The transition registers encode $|T|$ transitions and no transition state, thus requiring $\lceil\log_2(|T|+1)\rceil$ qubits. If $|T|$ is not a power of
two, some computational basis states do not correspond to any
transition. These states are treated as idle operations, leaving the
marking register unchanged. For example, when $|T|=3$, each
transition register contains two qubits, where $01$, $10$,
and $11$ to encode $T_1$, $T_2$, and $T_3$, respectively, while $00$
represents the no transition.} in
$T$.}

\vspace{2pt}
{\bf Step 2:}
{\rm Define $p$ one-qubit firing registers
$h_1,h_2,\ldots,h_p$, where $h_i= \ket{1}$  indicates that the  transition in $\alpha_i$  has fired and $h_i = \ket{0}$ indicates that it has not fired\footnote{The firing register along with transition register, enables the transition operator $U_T$ to be implemented as a reversible (unitary) quantum operation.},
$1\le i\le p$}.\\
\indent
{\bf Step 3:}
{\rm Encode the initial marking $\mu_0$ in the marking register\footnote{The marking register stores the number of $q$-tokens in each place, and if the QPN contains $n$ $q$-tokens, each place is encoded with
$\lceil\log_2(n+1)\rceil$ qubits, allowing the cardinality of each place
to be represented by the values $0,1,\ldots,n.$ Thus the marking register holds $m \lceil\log_2(n+1)\rceil$ qubits.}
$M$.}

\vspace{1pt}\noindent
{\bf for} {\rm $i = 1$ to $p$}\\
\{
{\bf Step 4:}
{\rm Apply Hadamard gates to all qubits of $\alpha_i$, creating a superposition
over all transitions, including the no transition.}

\vspace{2pt}\indent
{\bf Step 5:}
{\rm Controlled by the transition register $\alpha_i$,
apply the reversible transition operator\footnote{The
operator $U_T$ performs all transitions of the QPN as a reversible
quantum operation simultaneously to update the marking register
by moving the number of $q$-tokens from the input places to the output places according to the
directed labeled arcs of the QPN. For an individual transition $t$ in $T$, the operator $U_t$ is mathematically expressed as
\[
U_t|\mu\rangle|b\rangle
=
\left|\mu+e_t(\mu,v_t)\Delta_t\right\rangle
\left|b\oplus e_t(\mu,v_t)\right\rangle,
\]
where $|\mu\rangle$ represents the current marking, $e_t(\mu,v_t)$ indicates whether transition $t$ is enabled, $\Delta_t$ represents the change in the marking caused by firing $t$, and $|b\rangle$ is the corresponding firing qubit.} $U_T$ to the marking register
$M$ and the corresponding firing register $h_i$. }\\
\}

{\bf Step 6:}
{\rm Let $A$ denote the reversible  state-preparation operator (circuit)
that prepares the Phase-I quantum state.}\\

\vspace{-10pt}\noindent{\bf Phase II: Grover Amplitude Amplification}\\
\indent
{\bf Step 7:}
{\rm Define a phase oracle $O_{\mu_d}$ that compares the contents of the
marking register $M$ with the binary encoding of the target marking
$\mu_d$. If they are equal, the oracle flips the phase of the
corresponding computational basis state. Otherwise, the oracle acts as
the identity.\footnote{The oracle performs the comparison only on the
marking register. If multiple computational basis states encode the same
target marking $\mu_d$, the oracle flips the phase of each such basis
state, regardless of the values stored in the transition register and
firing register.}}

\noindent
{\bf repeat} $r$ times\footnote{$r$ denotes the number of Grover iterations and is expectedly given by $\pi/4 \sqrt{N/M},$ where
$N$ and $M$ represent the total number of basis states generated during Phase I and  the number of generated solutions, respectively.}\\
\{

{\bf Step\, 8:}
{\rm Apply the phase oracle $O_{\mu_d}$.}\\
\indent
{\bf Step 9:}
{\rm Apply the inverse Phase-I state-preparation operator
$A^\dagger$.}\\
\indent
{\bf Step 10:}
{\rm Apply the reflection operator\footnote{$S_0$  flips the phase of
the all-zero computational basis state of the entire quantum register,
including the marking register, transition registers, and
firing register, while leaving all
other computational basis states unchanged. Together with $A$ and
$A^\dagger$, it implements Grover's diffusion operator.}
$S_0$.}\\
\indent
{\bf Step 11:}
{\rm Apply the Phase-I state-preparation operator $A$.}\\
\}\\
\indent
{\bf Step 12:}
{\rm Measure the marking register and obtain a marking $\mu$.}

{\bf if} {\rm $\mu=\mu_d$}
{\rm return {\it  $\mu_d$ can be reached  from $\mu_0$ with high probability}.}

\vspace{-10pt}
\subsection{Algorithm Description}
\vspace{-2pt}\noindent 
As stated, the proposed quantum algorithm consists of two phases. In the first phase, the initial marking is first encoded in the marking register $M$, and then a superposition of reachable markings is generated by iteratively applying the transition operator $U_T$ to the marking register. In each iteration $i$ ($i = 1,2,\ldots,p$), a Hadamard gate is applied to the transition register $\alpha_i$ to prepare an equal superposition over all transitions in $T$ including the no transition.  Controlled by the transition register $\alpha_i$, the reversible transition operator $U_T$ determines whether the selected transition is enabled under the current marking. If it is enabled, $U_T$ updates the marking register by moving the number of $q$-tokens from the input places to the output places according to the directed labeled arcs, and updates the corresponding firing register $h_i$. Otherwise, the marking register remains unchanged. Repeating these steps for $p$ iterations prepares a quantum superposition of all reachable markings from the initial marking, together with their corresponding transitions and firing results. In each iteration, every computational branch advances by at most one transition by simultaneously exploring all enabled transition choices in
superposition. Therefore, the total number of iterations is upper bounded by the number of reachable markings, $L.$

In the second phase, Grover's amplitude amplification algorithm is applied to the superposed states prepared in Phase~I. A quantum oracle compares the marking register $M$ with the binary encoding of the target marking and flips the phase of every computational branch whose marking equals to the target marking, irrespective of the values stored in the transition-selection and firing registers. Each Grover iteration consists of applying the oracle, the inverse state-preparation operator $A^\dagger$, the reflection operator $S_0$, and the state-preparation operator $A$. These iterations amplify the amplitudes of the computational branches corresponding to the target marking while suppressing the amplitudes of the remaining branches. Finally, the marking register is measured to determine whether the target marking is reachable with high probability. 

\vspace{-15pt}
\subsection{Correctness Analysis}

\vspace{-2pt}\noindent
We now prove the correctness of  Algorithm 1.  We first establish
that the transition operator $U_T$ is unitary. Then we show that Phase I prepares a
superposition containing exactly the markings reachable from the initial
marking within the bound $p$. Finally, we prove that the amplitude
amplification phase correctly determines whether the target marking is
reachable with high probability.

\begin{lemma}
The transition operator $U_T$ is unitary.
\end{lemma}

\begin{proof}
Consider a superposition of markings, $M_i$  in an arbitrary iteration, and transition $\alpha_{t}.$ The transition of $M_i$  under $U_T$ is given by

\[
U_T: |M_i\rangle |\alpha_t\rangle |0\rangle_h
\xrightarrow{\,T_i\,}
|M'_i\rangle |\alpha_t\rangle |0\rangle_h.
\]

\noindent
where $M'$ is the marking obtained after firing transition $T_i$ in $\alpha_t$. If  $T_i$ in $\alpha_t$ does not fire then 
 \[U_T:
|M_i\rangle |\alpha_t\rangle |0\rangle_h
\xrightarrow{\,T_i\,}
|M_i\rangle |\alpha_t\rangle |\mathrm{0}\rangle_h.
\]

\noindent
In both cases, since the transition register $\alpha_t$  and the  firing qubit are preserved, marking register $M_i$ can be obtained from either $M_i'$ or $M_i$ Hence  $U_T$ is unitary. 
\end{proof}

\begin{theorem}
Algorithm 1 correctly decides whether the target marking $\mu_d$ is reachable
from $\mu_0$ within the bound $p$, with high probability.
\end{theorem}

\begin{proof}
It is easy to show that  after Phase I, the marking
register contains all and only markings reachable from
$\mu_0$ after $p$ iterations of firing. Suppose first that $\mu_d$ is reachable from $\mu_0$ within $p$ 
iterations in Phase-I. Then,  there must exist at least one basis state in
the Phase-I state whose marking register is equal to $\mu_d$. The oracle
flips the phase of exactly those basis states whose marking register is equal to
$\mu_d$. Now, by Lemma 1, $U_T$ is unitary, and therefore the Grover diffusion operator in Phase II
amplifies the amplitude of the target marking. After $r$
iterations, measuring the marking register returns $\mu_d$ with high
probability. Hence, the algorithm returns {\em Reachable} with high probability.
On the other hand, if $\mu_d$ is not reachable from $\mu_0$ within $p$
iterations of Phase I then no basis state in the Phase-I would have  a marking, which is equal to $\mu_d$. Therefore, the oracle
$O_{\mu_d}$ marks no basis state, and the amplitude of $\mu_d$ in the marking register remains zero.
Consequently, measuring the marking register cannot return $\mu_d$, and the
algorithm returns {\em Not Reachable}. Hence, Algorithm 1 correctly proves bounded reachability of $\mu_d$ from
$\mu_0$ within the bound $p$, with high probability.
\end{proof}

\vspace{-15pt}
\section{Example}

\noindent
In this section, we illustrate the proposed algorithm through a detailed example.
\begin{figure}[t]
\centering
\includegraphics[height=6 cm, width=\linewidth]{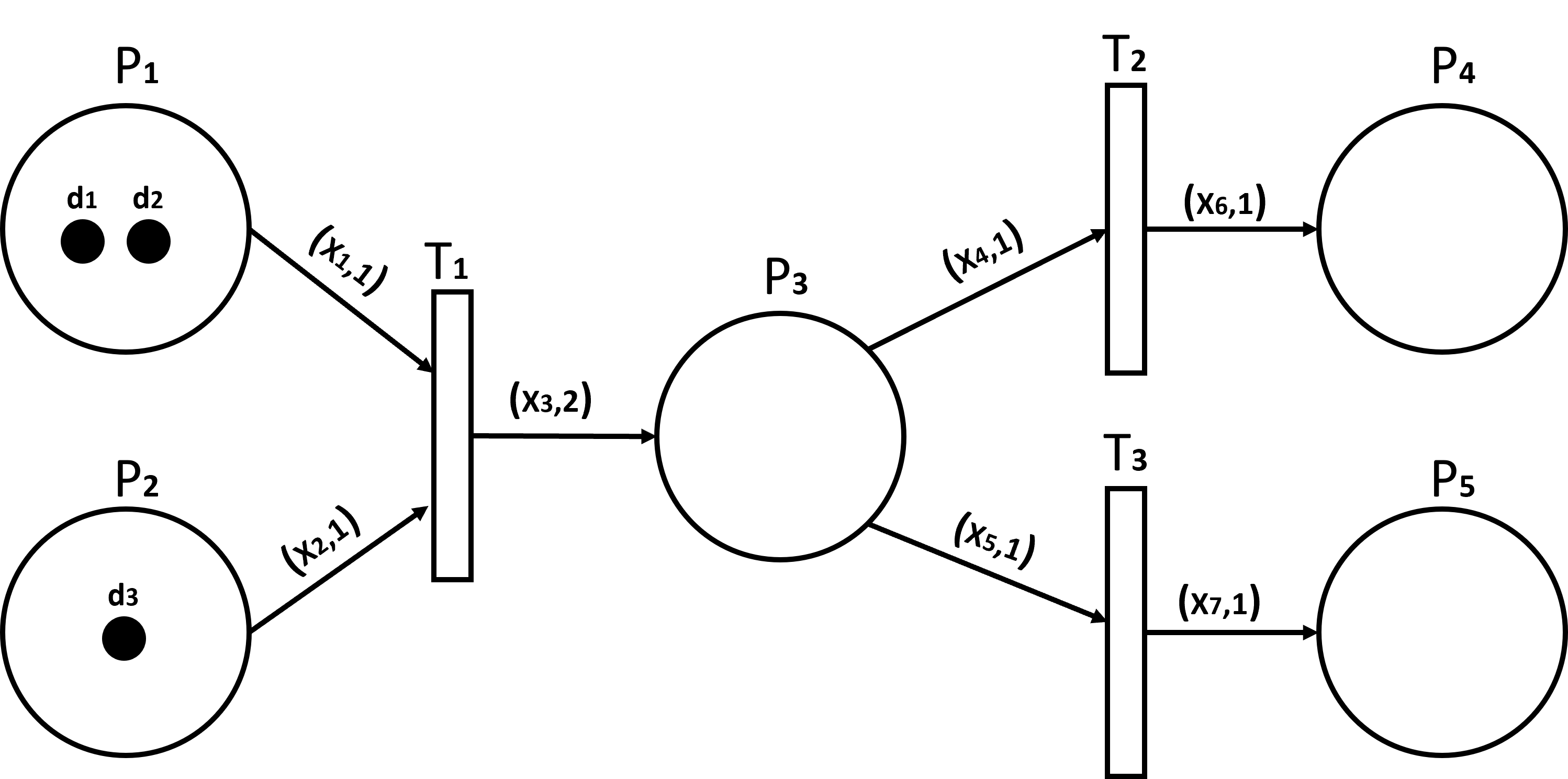}
\caption{Illustrative example of the proposed reachability algorithm, with $m=5$ places, $k=3$ transitions, and $n=3$ $q$-tokens.}
\label{fig:branching}
\end{figure} 
Consider the QPN in Figure~\ref{fig:branching}, which consists of
$n=3$ $q$-tokens $d_1,d_2,d_3$, $m=5$ places
$P_1,P_2,P_3,P_4,P_5$, $k=3$ transitions
$T_1,T_2,T_3$, and $z=7$ directed arcs
$x_1,x_2,\ldots,x_7$, with their firing capacities indicated next to
each arc. As QPN contains only $3$ $q$-tokens, each place may contain between $0$ and $3$
$q$-tokens, which can be represented by the two-bit encodings
$\ket{00}$, $\ket{01}$, $\ket{10}$, and $\ket{11}$, respectively. Therefore, the marking register consists of ten qubits, with two qubits allocated to each place.
Since the QPN contains three transitions, each transition
register $\alpha_i$
($i=1,2,\ldots,p$) consists of two qubits. The encodings $\ket{00}, \ket{01}$, $\ket{10}$, and $\ket{11}$ represent the selection of no transition, $T_1$, $T_2$, and $T_3$, respectively. For illustrative purposes, the algorithm is run for $p=3$ iterations. Consequently, three transition registers
$\alpha_1,\alpha_2,\alpha_3$ and three corresponding firing registers $h_1,h_2,h_3$ are used. 

\begin{table}[t]
\centering
\vspace{-5pt}
\footnotesize
\renewcommand{\arraystretch}{1.15}
\setlength{\tabcolsep}{2.5pt}
\setlength{\arrayrulewidth}{0.7pt}

\begin{tabular*}{\columnwidth}{@{\extracolsep{\fill}}|c|c|c|c|c|c|}
\hline
\textbf{Ancillary Registers} &
\multicolumn{5}{c|}{\textbf{Marking of Places}} \\

\begin{tabular}{@{}>{\centering\arraybackslash}c@{\hspace{8pt}}
                @{\hspace{8pt}}>{\centering\arraybackslash}c@{}}
$\alpha_1$ & $h_1$
\end{tabular}
&
\multicolumn{1}{c}{$P_1$} &
\multicolumn{1}{c}{$P_2$} &
\multicolumn{1}{c}{$P_3$} &
\multicolumn{1}{c}{$P_4$} &
\multicolumn{1}{c|}{$P_5$}
\\
\hline

{\renewcommand{\arraystretch}{1.15}
\begin{tabular}{@{}>{\centering\arraybackslash}c@{\hspace{6pt}}|
                @{\hspace{6pt}}>{\centering\arraybackslash}c@{}}
$\ket{00}$ & $\ket{0}$
\end{tabular}}
&
$\ket{10}$ &
$\ket{01}$ &
$\ket{00}$ &
$\ket{00}$ &
$\ket{00}$
\\
\hline

\end{tabular*}
\caption{Initial Quantum register configurations.}
\vspace{-15pt}
\label{tab:zero_iteration}
\end{table}

\begin{table}[t]
\centering

\footnotesize
\renewcommand{\arraystretch}{1.15}
\setlength{\tabcolsep}{2.5pt}
\setlength{\arrayrulewidth}{0.7pt}

\begin{tabular*}{\columnwidth}{@{\extracolsep{\fill}}|c|c|c|c|c|c|}
\hline
\textbf{Ancillary Registers} &
\multicolumn{5}{c|}{\textbf{Marking of Places}} \\

\begin{tabular}{@{}>{\centering\arraybackslash}c@{\hspace{8pt}}
                @{\hspace{8pt}}>{\centering\arraybackslash}c@{}}
$\alpha_1$ & $h_1$
\end{tabular}
&
\multicolumn{1}{c}{$P_1$} &
\multicolumn{1}{c}{$P_2$} &
\multicolumn{1}{c}{$P_3$} &
\multicolumn{1}{c}{$P_4$} &
\multicolumn{1}{c|}{$P_5$}
\\
\hline

{\renewcommand{\arraystretch}{1.15}
\begin{tabular}{@{}>{\centering\arraybackslash}c@{\hspace{6pt}}|
                @{\hspace{6pt}}>{\centering\arraybackslash}c@{}}
$\ket{00}$ & $\ket{0}$\\
$\ket{10}$ & $\ket{0}$\\
$\ket{11}$ & $\ket{0}$
\end{tabular}}
&
$\ket{10}$ &
$\ket{01}$ &
$\ket{00}$ &
$\ket{00}$ &
$\ket{00}$
\\
\hline

{\renewcommand{\arraystretch}{1.15}
\begin{tabular}{@{}>{\centering\arraybackslash}c@{\hspace{6pt}}|
                @{\hspace{6pt}}>{\centering\arraybackslash}c@{}}
$\ket{01}$ & $\ket{1}$
\end{tabular}}
&
$\ket{01}$ &
$\ket{00}$ &
$\ket{10}$ &
$\ket{00}$ &
$\ket{00}$
\\
\hline

\end{tabular*}
\caption{Quantum register configurations after the first transition-selection stage.}
\label{tab:first_iteration}
\end{table}
\noindent 
Table~\ref{tab:zero_iteration} shows the initial configuration of the
quantum registers, including the marking and ancillary registers. The markings are represented by their cardinality vectors. For example, the vector $(P_1,P_2,P_3,P_4,P_5)=(\ket{10},\ket{01},\ket{00},\ket{00},\ket{00})$ indicates that two $q$-token is located in place $P_1$, one $q$-tokens are located in place $P_2$, and no $q$-tokens are present in the remaining places.\\
\indent
In the first iteration, Hadamard gates are applied to the qubits of the transition register $\alpha_1$, creating an equal superposition over the four possible selections, i.e., $\ket{00}$ for the no transition, $\ket{01}$ for $T_1$, $\ket{10}$ for $T_2$, and $\ket{11}$ for
$T_3$. The reversible transition operator $U_T$ is then applied to fire the
selected transition, if it is enabled under the current marking. Given that the initial marking register is $M =\ket{10,01,00,00,00}$ and $T_1$ is connected to $P_1$ and $P_2,$ only transition $T_1$ is enabled. Consequently,  in the  basis state, where $\alpha_1=\ket{01}$, and the operator $U_T$ fires $T_1$, moving the $q$-tokens $d_1$ and $d_3$ from places $P_1$ and $P_2$ to place $P_3$, and stores the transition fired in the firing register $h_1$. For the rest of the basis states, where $\alpha_i \in \{\ket{00},\ket{10},\ket{11}\},$ the marking register and the firing register remain unchanged because no transition is enabled. The resulting quantum register configurations are shown in Table~\ref{tab:first_iteration}.

\indent
In the second iteration, Hadamard gates are applied to the qubits of the transition-selection register $\alpha_2$, creating an equal superposition over the four possible selections i.e. $\ket{00}$ for the no transition, $\ket{01}$ for $T_1$, $\ket{10}$ for $T_2$, and $\ket{11}$ for
$T_3$. The reversible transition operator $U_T$ is then applied to fire the
selected transition, if it is enabled under the current marking.The branch with marking register
$M=\ket{10,01,00,00,00}$ evolves into
$M=\ket{01,00,10,00,00}$.
Subsequently, the branch with marking register
$M=\ket{01,00,10,00,00}$ evolves into the branches
$M=\ket{01,00,01,01,00}$ and
$M=\ket{01,00,01,00,01}$. In contrast, selecting the no transition or a transition that is not
enabled leaves both the marking and firing registers unchanged.  The resulting quantum register configurations after the second iteration are shown in Table~\ref{tab:second_iteration}. 

\begin{table}[t]
\centering
\caption{Quantum register configurations after the second transition-selection stage.}
\label{tab:second_iteration}
\footnotesize
\renewcommand{\arraystretch}{1.15}
\setlength{\tabcolsep}{2.5pt}
\setlength{\arrayrulewidth}{0.7pt}

\begin{tabular*}{\columnwidth}{@{\extracolsep{\fill}}|c|c|c|c|c|c|}
\hline
\textbf{Ancillary Registers} &
\multicolumn{5}{c|}{\textbf{Marking of Places}} \\

\begin{tabular}{@{}>{\centering\arraybackslash}c@{\hspace{8pt}}
                @{\hspace{8pt}}>{\centering\arraybackslash}c@{\hspace{8pt}}
                @{\hspace{8pt}}>{\centering\arraybackslash}c@{\hspace{8pt}}
                @{\hspace{8pt}}>{\centering\arraybackslash}c@{}}
$\alpha_1$ & $h_1$ & $\alpha_2$ & $h_2$
\end{tabular}
&
\multicolumn{1}{c}{$P_1$} &
\multicolumn{1}{c}{$P_2$} &
\multicolumn{1}{c}{$P_3$} &
\multicolumn{1}{c}{$P_4$} &
\multicolumn{1}{c|}{$P_5$}
\\
\hline

{\renewcommand{\arraystretch}{1.15}
\begin{tabular}{@{}>{\centering\arraybackslash}c@{\hspace{6pt}}|
                @{\hspace{6pt}}>{\centering\arraybackslash}c@{\hspace{6pt}}|
                @{\hspace{6pt}}>{\centering\arraybackslash}c@{\hspace{6pt}}|
                @{\hspace{6pt}}>{\centering\arraybackslash}c@{}}
$\ket{00}$ & $\ket{0}$ & $\ket{00}$ & $\ket{0}$\\
$\ket{00}$ & $\ket{0}$ & $\ket{10}$ & $\ket{0}$\\
$\ket{00}$ & $\ket{0}$ & $\ket{11}$ & $\ket{0}$\\
$\ket{10}$ & $\ket{0}$ & $\ket{00}$ & $\ket{0}$\\
$\ket{10}$ & $\ket{0}$ & $\ket{10}$ & $\ket{0}$\\
$\ket{10}$ & $\ket{0}$ & $\ket{11}$ & $\ket{0}$\\
$\ket{11}$ & $\ket{0}$ & $\ket{00}$ & $\ket{0}$\\
$\ket{11}$ & $\ket{0}$ & $\ket{10}$ & $\ket{0}$\\
$\ket{11}$ & $\ket{0}$ & $\ket{11}$ & $\ket{0}$
\end{tabular}}
&
$\ket{10}$ & $\ket{01}$ & $\ket{00}$ & $\ket{00}$ & $\ket{00}$
\\
\hline

{\renewcommand{\arraystretch}{1.15}
\begin{tabular}{@{}>{\centering\arraybackslash}c@{\hspace{6pt}}|
                @{\hspace{6pt}}>{\centering\arraybackslash}c@{\hspace{6pt}}|
                @{\hspace{6pt}}>{\centering\arraybackslash}c@{\hspace{6pt}}|
                @{\hspace{6pt}}>{\centering\arraybackslash}c@{}}
$\ket{00}$ & $\ket{0}$ & $\ket{01}$ & $\ket{1}$\\
$\ket{01}$ & $\ket{1}$ & $\ket{00}$ & $\ket{0}$\\
$\ket{01}$ & $\ket{1}$ & $\ket{01}$ & $\ket{0}$\\
$\ket{10}$ & $\ket{0}$ & $\ket{01}$ & $\ket{1}$\\
$\ket{11}$ & $\ket{0}$ & $\ket{01}$ & $\ket{1}$
\end{tabular}}
&
$\ket{01}$ & $\ket{00}$ & $\ket{10}$ & $\ket{00}$ & $\ket{00}$
\\
\hline

{\renewcommand{\arraystretch}{1.15}
\begin{tabular}{@{}>{\centering\arraybackslash}c@{\hspace{6pt}}|
                @{\hspace{6pt}}>{\centering\arraybackslash}c@{\hspace{6pt}}|
                @{\hspace{6pt}}>{\centering\arraybackslash}c@{\hspace{6pt}}|
                @{\hspace{6pt}}>{\centering\arraybackslash}c@{}}
$\ket{01}$ & $\ket{1}$ & $\ket{10}$ & $\ket{1}$
\end{tabular}}
&
$\ket{01}$ & $\ket{00}$ & $\ket{01}$ & $\ket{01}$ & $\ket{00}$
\\
\hline

{\renewcommand{\arraystretch}{1.15}
\begin{tabular}{@{}>{\centering\arraybackslash}c@{\hspace{6pt}}|
                @{\hspace{6pt}}>{\centering\arraybackslash}c@{\hspace{6pt}}|
                @{\hspace{6pt}}>{\centering\arraybackslash}c@{\hspace{6pt}}|
                @{\hspace{6pt}}>{\centering\arraybackslash}c@{}}
$\ket{01}$ & $\ket{1}$ & $\ket{11}$ & $\ket{1}$
\end{tabular}}
&
$\ket{01}$ & $\ket{00}$ & $\ket{01}$ & $\ket{00}$ & $\ket{01}$
\\
\hline

\end{tabular*}
\vspace{-15pt}
\end{table}

\indent
In the third iteration, Hadamard gates are applied to the qubits of the transition-selection register $\alpha_3$, creating an equal superposition over the four possible selections i.e.  $\ket{00}$ for the no transition, $\ket{01}$ for $T_1$, $\ket{10}$ for $T_2$, and $\ket{11}$ for
$T_3$. The reversible transition operator $U_T$ is then applied to fire the
selected transition, if it is enabled under the current marking. he branch with marking register
$M=\ket{10,01,00,00,00}$ evolves into
$M=\ket{01,00,10,00,00}$. The branch with marking register
$M=\ket{01,00,10,00,00}$ evolves into the branches
$M=\ket{01,00,01,01,00}$ and
$M=\ket{01,00,01,00,01}$.
The branch with marking register
$M=\ket{01,00,01,01,00}$ further evolves into
$M=\ket{01,00,00,10,00}$ and
$M=\ket{01,00,00,01,01}$.
Similarly, the branch with marking register
$M=\ket{01,00,01,00,01}$ evolves into
$M=\ket{01,00,00,00,10}$ and
$M=\ket{01,00,00,01,01}$. In contrast, selecting the no transition or a transition that is not
enabled leaves both the marking and firing registers unchanged. The resulting quantum register configurations after the third iteration are shown in Table~\ref{tab:third_iteration}.

After the third iteration, Phase-I terminates. Consequently, the quantum
register contains a superposition of all reachable markings generated in $p$ iterations ($p = 3$ in the example), together with the corresponding transition and firing registers. 

In Phase II, Grover's amplitude amplification algorithm is applied to the quantum register prepared in Phase-I. Let $A$ denote the reversible Phase-I state-preparation operator that generated Phase-1 states. A phase oracle $O_{\mu_d}$ is constructed to identify the desired target marking $\mu_d$ by comparing the marking register $M$ with the binary encoding of $\mu_d$ and flipping the phase of every computational branch whose marking equals the target marking.  
\begin{table}[H]
\centering
\caption{Quantum register configurations after the third transition-selection stage.}
\label{tab:third_iteration}
\scriptsize
\renewcommand{\arraystretch}{1.15}
\setlength{\tabcolsep}{2.0pt}
\setlength{\arrayrulewidth}{0.7pt}

\begin{tabular*}{\columnwidth}{@{\extracolsep{\fill}}|c|c|c|c|c|c|}
\hline
\textbf{Ancillary Registers} &
\multicolumn{5}{c|}{\textbf{Marking of Places}}  \\

\begin{tabular}{@{}>{\centering\arraybackslash}c@{\hspace{5pt}}
                @{\hspace{5pt}}>{\centering\arraybackslash}c@{\hspace{5pt}}
                @{\hspace{5pt}}>{\centering\arraybackslash}c@{\hspace{5pt}}
                @{\hspace{5pt}}>{\centering\arraybackslash}c@{\hspace{5pt}}
                @{\hspace{5pt}}>{\centering\arraybackslash}c@{\hspace{5pt}}
                @{\hspace{5pt}}>{\centering\arraybackslash}c@{}}
$\alpha_1$ & $h_1$ & $\alpha_2$ & $h_2$ & $\alpha_3$ & $h_3$
\end{tabular}
&
\multicolumn{1}{c}{$P_1$} &
\multicolumn{1}{c}{$P_2$} &
\multicolumn{1}{c}{$P_3$} &
\multicolumn{1}{c}{$P_4$} &
\multicolumn{1}{c|}{$P_5$}
\\
\hline

{\renewcommand{\arraystretch}{1.15}
\begin{tabular}{@{}>{\centering\arraybackslash}c@{\hspace{3pt}}|
                @{\hspace{3pt}}>{\centering\arraybackslash}c@{\hspace{3pt}}|
                @{\hspace{3pt}}>{\centering\arraybackslash}c@{\hspace{3pt}}|
                @{\hspace{3pt}}>{\centering\arraybackslash}c@{\hspace{3pt}}|
                @{\hspace{3pt}}>{\centering\arraybackslash}c@{\hspace{3pt}}|
                @{\hspace{3pt}}>{\centering\arraybackslash}c@{}}
$\ket{00}$ & $\ket{0}$ & $\ket{00}$ & $\ket{0}$ & $\ket{00}$ & $\ket{0}$\\
$\ket{00}$ & $\ket{0}$ & $\ket{00}$ & $\ket{0}$ & $\ket{10}$ & $\ket{0}$\\
$\ket{00}$ & $\ket{0}$ & $\ket{00}$ & $\ket{0}$ & $\ket{11}$ & $\ket{0}$\\
$\ket{00}$ & $\ket{0}$ & $\ket{10}$ & $\ket{0}$ & $\ket{00}$ & $\ket{0}$\\
$\ket{00}$ & $\ket{0}$ & $\ket{10}$ & $\ket{0}$ & $\ket{10}$ & $\ket{0}$\\
$\ket{00}$ & $\ket{0}$ & $\ket{10}$ & $\ket{0}$ & $\ket{11}$ & $\ket{0}$\\
$\ket{00}$ & $\ket{0}$ & $\ket{11}$ & $\ket{0}$ & $\ket{00}$ & $\ket{0}$\\
$\ket{00}$ & $\ket{0}$ & $\ket{11}$ & $\ket{0}$ & $\ket{10}$ & $\ket{0}$\\
$\ket{00}$ & $\ket{0}$ & $\ket{11}$ & $\ket{0}$ & $\ket{11}$ & $\ket{0}$\\
$\ket{10}$ & $\ket{0}$ & $\ket{00}$ & $\ket{0}$ & $\ket{00}$ & $\ket{0}$\\
$\ket{10}$ & $\ket{0}$ & $\ket{00}$ & $\ket{0}$ & $\ket{10}$ & $\ket{0}$\\
$\ket{10}$ & $\ket{0}$ & $\ket{00}$ & $\ket{0}$ & $\ket{11}$ & $\ket{0}$\\
$\ket{10}$ & $\ket{0}$ & $\ket{10}$ & $\ket{0}$ & $\ket{00}$ & $\ket{0}$\\
$\ket{10}$ & $\ket{0}$ & $\ket{10}$ & $\ket{0}$ & $\ket{10}$ & $\ket{0}$\\
$\ket{10}$ & $\ket{0}$ & $\ket{10}$ & $\ket{0}$ & $\ket{11}$ & $\ket{0}$\\
$\ket{10}$ & $\ket{0}$ & $\ket{11}$ & $\ket{0}$ & $\ket{00}$ & $\ket{0}$\\
$\ket{10}$ & $\ket{0}$ & $\ket{11}$ & $\ket{0}$ & $\ket{10}$ & $\ket{0}$\\
$\ket{10}$ & $\ket{0}$ & $\ket{11}$ & $\ket{0}$ & $\ket{11}$ & $\ket{0}$\\
$\ket{11}$ & $\ket{0}$ & $\ket{00}$ & $\ket{0}$ & $\ket{00}$ & $\ket{0}$\\
$\ket{11}$ & $\ket{0}$ & $\ket{00}$ & $\ket{0}$ & $\ket{10}$ & $\ket{0}$\\
$\ket{11}$ & $\ket{0}$ & $\ket{00}$ & $\ket{0}$ & $\ket{11}$ & $\ket{0}$\\
$\ket{11}$ & $\ket{0}$ & $\ket{10}$ & $\ket{0}$ & $\ket{00}$ & $\ket{0}$\\
$\ket{11}$ & $\ket{0}$ & $\ket{10}$ & $\ket{0}$ & $\ket{10}$ & $\ket{0}$\\
$\ket{11}$ & $\ket{0}$ & $\ket{10}$ & $\ket{0}$ & $\ket{11}$ & $\ket{0}$\\
$\ket{11}$ & $\ket{0}$ & $\ket{11}$ & $\ket{0}$ & $\ket{00}$ & $\ket{0}$\\
$\ket{11}$ & $\ket{0}$ & $\ket{11}$ & $\ket{0}$ & $\ket{10}$ & $\ket{0}$\\
$\ket{11}$ & $\ket{0}$ & $\ket{11}$ & $\ket{0}$ & $\ket{11}$ & $\ket{0}$
\end{tabular}}
&
$\ket{10}$ & $\ket{01}$ & $\ket{00}$ & $\ket{00}$ & $\ket{00}$
\\
\hline

{\renewcommand{\arraystretch}{1.15}
\begin{tabular}{@{}>{\centering\arraybackslash}c@{\hspace{3pt}}|
                @{\hspace{3pt}}>{\centering\arraybackslash}c@{\hspace{3pt}}|
                @{\hspace{3pt}}>{\centering\arraybackslash}c@{\hspace{3pt}}|
                @{\hspace{3pt}}>{\centering\arraybackslash}c@{\hspace{3pt}}|
                @{\hspace{3pt}}>{\centering\arraybackslash}c@{\hspace{3pt}}|
                @{\hspace{3pt}}>{\centering\arraybackslash}c@{}}
$\ket{00}$ & $\ket{0}$ & $\ket{00}$ & $\ket{0}$ & $\ket{01}$ & $\ket{1}$\\
$\ket{00}$ & $\ket{0}$ & $\ket{01}$ & $\ket{1}$ & $\ket{00}$ & $\ket{0}$\\
$\ket{00}$ & $\ket{0}$ & $\ket{01}$ & $\ket{1}$ & $\ket{01}$ & $\ket{0}$\\
$\ket{00}$ & $\ket{0}$ & $\ket{10}$ & $\ket{0}$ & $\ket{01}$ & $\ket{1}$\\
$\ket{00}$ & $\ket{0}$ & $\ket{11}$ & $\ket{0}$ & $\ket{01}$ & $\ket{1}$\\
$\ket{01}$ & $\ket{1}$ & $\ket{00}$ & $\ket{0}$ & $\ket{00}$ & $\ket{0}$\\
$\ket{01}$ & $\ket{1}$ & $\ket{00}$ & $\ket{0}$ & $\ket{01}$ & $\ket{0}$\\
$\ket{01}$ & $\ket{1}$ & $\ket{01}$ & $\ket{0}$ & $\ket{00}$ & $\ket{0}$\\
$\ket{01}$ & $\ket{1}$ & $\ket{01}$ & $\ket{0}$ & $\ket{01}$ & $\ket{0}$\\
$\ket{10}$ & $\ket{0}$ & $\ket{00}$ & $\ket{0}$ & $\ket{01}$ & $\ket{1}$\\
$\ket{10}$ & $\ket{0}$ & $\ket{01}$ & $\ket{1}$ & $\ket{00}$ & $\ket{0}$\\
$\ket{10}$ & $\ket{0}$ & $\ket{01}$ & $\ket{1}$ & $\ket{01}$ & $\ket{0}$\\
$\ket{10}$ & $\ket{0}$ & $\ket{10}$ & $\ket{0}$ & $\ket{01}$ & $\ket{1}$\\
$\ket{10}$ & $\ket{0}$ & $\ket{11}$ & $\ket{0}$ & $\ket{01}$ & $\ket{1}$\\
$\ket{11}$ & $\ket{0}$ & $\ket{00}$ & $\ket{0}$ & $\ket{01}$ & $\ket{1}$\\
$\ket{11}$ & $\ket{0}$ & $\ket{01}$ & $\ket{1}$ & $\ket{00}$ & $\ket{0}$\\
$\ket{11}$ & $\ket{0}$ & $\ket{01}$ & $\ket{1}$ & $\ket{01}$ & $\ket{0}$\\
$\ket{11}$ & $\ket{0}$ & $\ket{10}$ & $\ket{0}$ & $\ket{01}$ & $\ket{1}$\\
$\ket{11}$ & $\ket{0}$ & $\ket{11}$ & $\ket{0}$ & $\ket{01}$ & $\ket{1}$
\end{tabular}}
&
$\ket{01}$ & $\ket{00}$ & $\ket{10}$ & $\ket{00}$ & $\ket{00}$
\\
\hline

{\renewcommand{\arraystretch}{1.15}
\begin{tabular}{@{}>{\centering\arraybackslash}c@{\hspace{3pt}}|
                @{\hspace{3pt}}>{\centering\arraybackslash}c@{\hspace{3pt}}|
                @{\hspace{3pt}}>{\centering\arraybackslash}c@{\hspace{3pt}}|
                @{\hspace{3pt}}>{\centering\arraybackslash}c@{\hspace{3pt}}|
                @{\hspace{3pt}}>{\centering\arraybackslash}c@{\hspace{3pt}}|
                @{\hspace{3pt}}>{\centering\arraybackslash}c@{}}
$\ket{00}$ & $\ket{0}$ & $\ket{01}$ & $\ket{1}$ & $\ket{10}$ & $\ket{1}$\\
$\ket{01}$ & $\ket{1}$ & $\ket{00}$ & $\ket{0}$ & $\ket{10}$ & $\ket{1}$\\
$\ket{01}$ & $\ket{1}$ & $\ket{01}$ & $\ket{0}$ & $\ket{10}$ & $\ket{1}$\\
$\ket{01}$ & $\ket{1}$ & $\ket{10}$ & $\ket{1}$ & $\ket{00}$ & $\ket{0}$\\
$\ket{01}$ & $\ket{1}$ & $\ket{10}$ & $\ket{1}$ & $\ket{01}$ & $\ket{0}$\\
$\ket{10}$ & $\ket{0}$ & $\ket{01}$ & $\ket{1}$ & $\ket{10}$ & $\ket{1}$\\
$\ket{11}$ & $\ket{0}$ & $\ket{01}$ & $\ket{1}$ & $\ket{10}$ & $\ket{1}$
\end{tabular}}
&
$\ket{01}$ & $\ket{00}$ & $\ket{01}$ & $\ket{01}$ & $\ket{00}$
\\
\hline

{\renewcommand{\arraystretch}{1.15}
\begin{tabular}{@{}>{\centering\arraybackslash}c@{\hspace{3pt}}|
                @{\hspace{3pt}}>{\centering\arraybackslash}c@{\hspace{3pt}}|
                @{\hspace{3pt}}>{\centering\arraybackslash}c@{\hspace{3pt}}|
                @{\hspace{3pt}}>{\centering\arraybackslash}c@{\hspace{3pt}}|
                @{\hspace{3pt}}>{\centering\arraybackslash}c@{\hspace{3pt}}|
                @{\hspace{3pt}}>{\centering\arraybackslash}c@{}}
$\ket{00}$ & $\ket{0}$ & $\ket{01}$ & $\ket{1}$ & $\ket{11}$ & $\ket{1}$\\
$\ket{01}$ & $\ket{1}$ & $\ket{00}$ & $\ket{0}$ & $\ket{11}$ & $\ket{1}$\\
$\ket{01}$ & $\ket{1}$ & $\ket{01}$ & $\ket{0}$ & $\ket{11}$ & $\ket{1}$\\
$\ket{01}$ & $\ket{1}$ & $\ket{11}$ & $\ket{1}$ & $\ket{00}$ & $\ket{0}$\\
$\ket{01}$ & $\ket{1}$ & $\ket{11}$ & $\ket{1}$ & $\ket{01}$ & $\ket{0}$\\
$\ket{10}$ & $\ket{0}$ & $\ket{01}$ & $\ket{1}$ & $\ket{11}$ & $\ket{1}$\\
$\ket{11}$ & $\ket{0}$ & $\ket{01}$ & $\ket{1}$ & $\ket{11}$ & $\ket{1}$
\end{tabular}}
&
$\ket{01}$ & $\ket{00}$ & $\ket{01}$ & $\ket{00}$ & $\ket{01}$
\\
\hline

{\renewcommand{\arraystretch}{1.15}
\begin{tabular}{@{}>{\centering\arraybackslash}c@{\hspace{3pt}}|
                @{\hspace{3pt}}>{\centering\arraybackslash}c@{\hspace{3pt}}|
                @{\hspace{3pt}}>{\centering\arraybackslash}c@{\hspace{3pt}}|
                @{\hspace{3pt}}>{\centering\arraybackslash}c@{\hspace{3pt}}|
                @{\hspace{3pt}}>{\centering\arraybackslash}c@{\hspace{3pt}}|
                @{\hspace{3pt}}>{\centering\arraybackslash}c@{}}
$\ket{01}$ & $\ket{1}$ & $\ket{10}$ & $\ket{1}$ & $\ket{10}$ & $\ket{1}$
\end{tabular}}
&
$\ket{01}$ & $\ket{00}$ & $\ket{00}$ & $\ket{10}$ & $\ket{00}$
\\
\hline

{\renewcommand{\arraystretch}{1.15}
\begin{tabular}{@{}>{\centering\arraybackslash}c@{\hspace{3pt}}|
                @{\hspace{3pt}}>{\centering\arraybackslash}c@{\hspace{3pt}}|
                @{\hspace{3pt}}>{\centering\arraybackslash}c@{\hspace{3pt}}|
                @{\hspace{3pt}}>{\centering\arraybackslash}c@{\hspace{3pt}}|
                @{\hspace{3pt}}>{\centering\arraybackslash}c@{\hspace{3pt}}|
                @{\hspace{3pt}}>{\centering\arraybackslash}c@{}}
$\ket{01}$ & $\ket{1}$ & $\ket{10}$ & $\ket{1}$ & $\ket{11}$ & $\ket{1}$\\
$\ket{01}$ & $\ket{1}$ & $\ket{11}$ & $\ket{1}$ & $\ket{10}$ & $\ket{1}$
\end{tabular}}
&
$\ket{01}$ & $\ket{00}$ & $\ket{00}$ & $\ket{01}$ & $\ket{01}$
\\
\hline

{\renewcommand{\arraystretch}{1.15}
\begin{tabular}{@{}>{\centering\arraybackslash}c@{\hspace{3pt}}|
                @{\hspace{3pt}}>{\centering\arraybackslash}c@{\hspace{3pt}}|
                @{\hspace{3pt}}>{\centering\arraybackslash}c@{\hspace{3pt}}|
                @{\hspace{3pt}}>{\centering\arraybackslash}c@{\hspace{3pt}}|
                @{\hspace{3pt}}>{\centering\arraybackslash}c@{\hspace{3pt}}|
                @{\hspace{3pt}}>{\centering\arraybackslash}c@{}}
$\ket{01}$ & $\ket{1}$ & $\ket{11}$ & $\ket{1}$ & $\ket{11}$ & $\ket{1}$
\end{tabular}}
&
$\ket{01}$ & $\ket{00}$ & $\ket{00}$ & $\ket{00}$ & $\ket{10}$
\\
\hline

\end{tabular*}
\end{table}

\indent
Grover iterations are then performed $r$ times, where each iteration consists of applying the oracle $O_{\mu_d}$, followed by the inverse state-preparation operator $A^\dagger$, the reflection operator $S_0$, and the state-preparation operator $A$. Consequently, the amplitudes of the computational branches corresponding to the target marking are amplified, while those of the remaining branches are suppressed. Finally, the marking register is measured. If the measured marking equals the target marking $\mu_d$, the algorithm reports \emph{Reachable}  with high probability.
\begin{figure*}[t]
\centering

% First row
\begin{subfigure}[b]{0.32\textwidth}
\centering
\includegraphics[width=\linewidth]{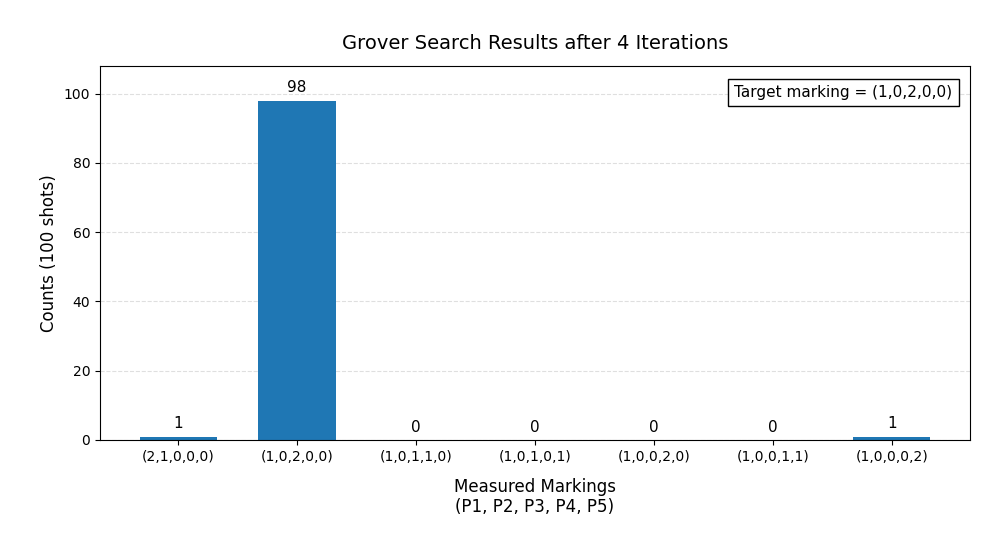}
\caption{Measurement outcomes for the reachable target marking
$(P_1,P_2,P_3,P_4,P_5)=(01,00,10,00,00)$.}
\end{subfigure}
\hfill
\begin{subfigure}[b]{0.32\textwidth}
\centering
\includegraphics[width=\linewidth]{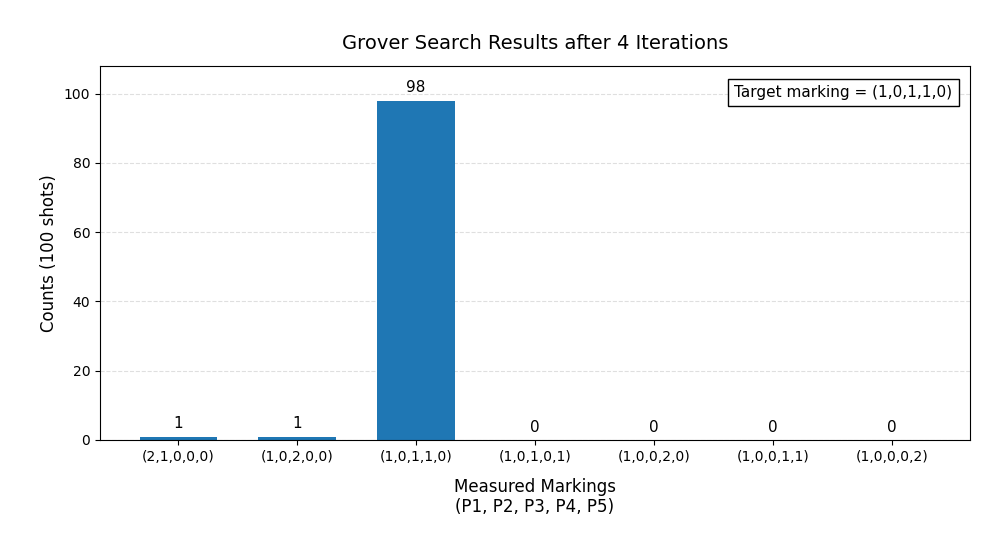}
\caption{Measurement outcomes for the reachable target marking
$(P_1,P_2,P_3,P_4,P_5)=(01,00,01,01,00)$.}
\end{subfigure}
\hfill
\begin{subfigure}[b]{0.32\textwidth}
\centering
\includegraphics[width=\linewidth]{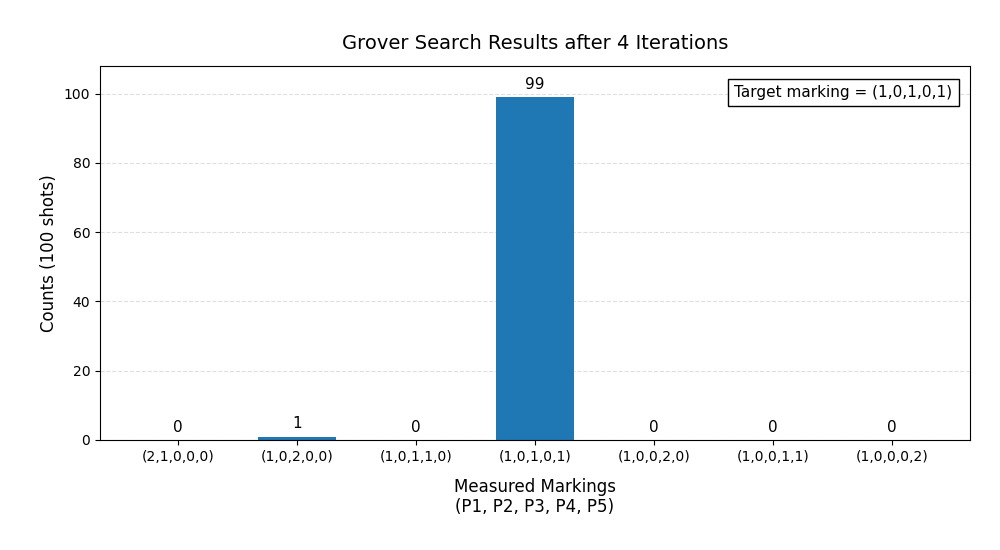}
\caption{Measurement outcomes for the reachable target marking
$(P_1,P_2,P_3,P_4,P_5)=(01,00,01,00,01)$.}
\end{subfigure}

\vspace{3mm}

% Second row
\begin{subfigure}[b]{0.32\textwidth}
\centering
\includegraphics[width=\linewidth]{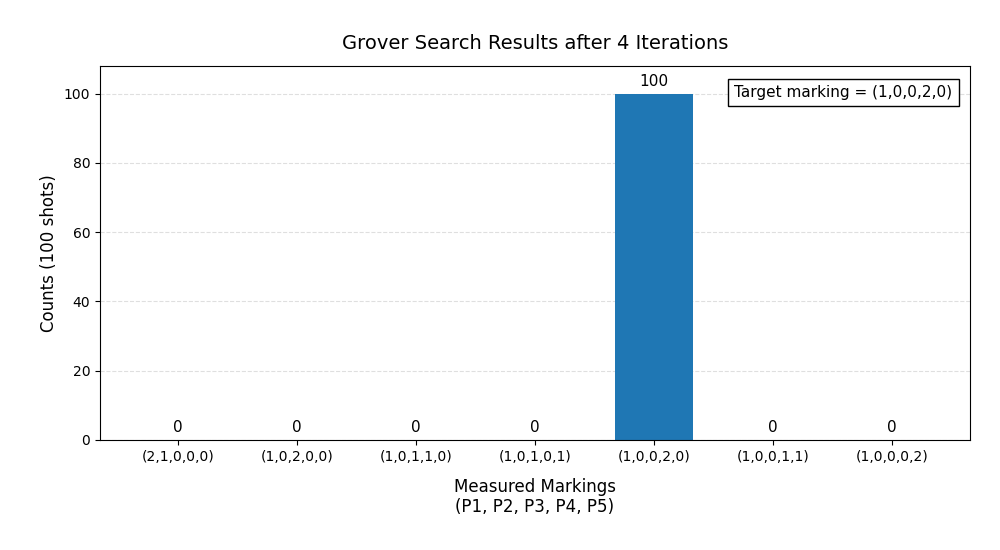}
\caption{Measurement outcomes for the reachable target marking
$(P_1,P_2,P_3,P_4,P_5)=(01,00,00,10,00)$.}
\end{subfigure}
\hfill
\begin{subfigure}[b]{0.32\textwidth}
\centering
\includegraphics[width=\linewidth]{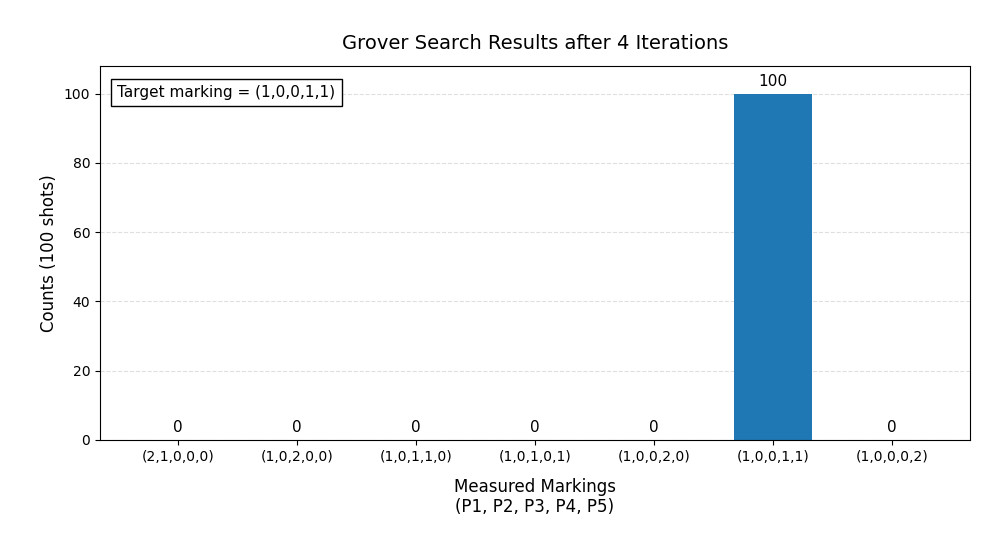}
\caption{Measurement outcomes for the reachable target marking
$(P_1,P_2,P_3,P_4,P_5)=(01,00,00,01,01)$.}
\end{subfigure}
\hfill
\begin{subfigure}[b]{0.32\textwidth}
\centering
\includegraphics[width=\linewidth]{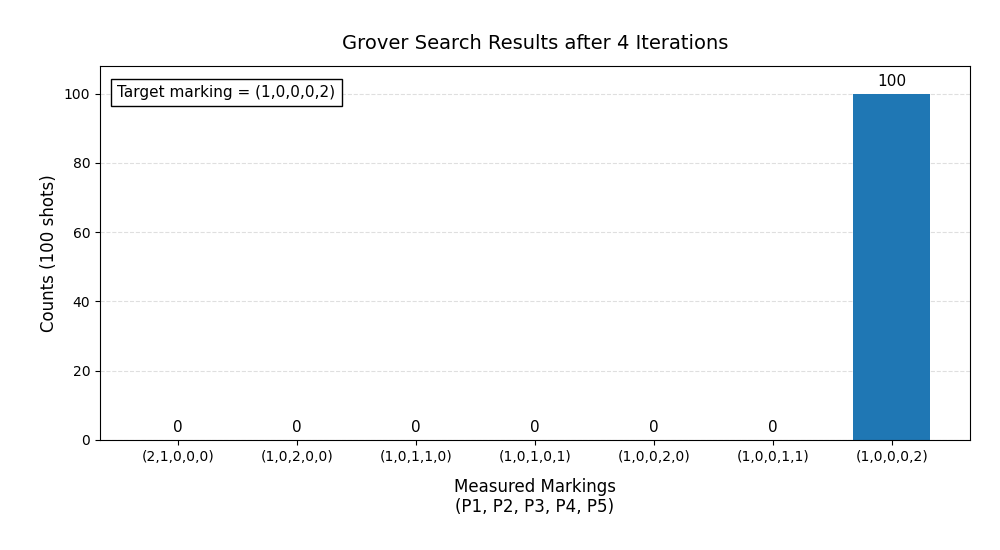}
\caption{Measurement outcomes for the reachable target marking
$(P_1,P_2,P_3,P_4,P_5)=(01,00,00,00,10)$.}
\end{subfigure}

\caption{Grover search results for reachable target markings.}
\label{fig:grover_targets}
\end{figure*}
\vspace{ 5pt}\\
\textbf{Experimental Evaluation:}
To validate the proposed reachability algorithm, a $Q\# $ implementation was developed for the QPN shown in Figure~\ref{fig:branching}, and a series of experiments was conducted. The implementation realizes both phases of the algorithm. The reachable states  in Table~\ref{tab:third_iteration}   correspond to all markings reachable from the initial marking to the end of $p=3$ iterations, in which Grover's algorithm searches for the desired target marking. Consequently, only markings contained in this superposition can be amplified during the Grover amplitude amplification phase. On the other hand, in the experiments we performed, measurement outcomes were collected over 100 shots for a subset of both reachable and non-reachable target markings. For the reachable cases, all reachable markings except the initial marking $(2,1,0,0,0)$ were included in the experiment. The resulting measurement distributions are shown in Figures~\ref{fig:grover_targets} and~\ref{fig:grover_utargets}. 

For the reachable target markings, the desired marking consistently dominates the measurement outcomes after a few Grover iterations. In particular, the target markings $(01,00,10,00,00)$, $(01,00,01,01,00)$, $(01,00,01,00,01)$, $(01,00,00,10,00)$, $(01,00,00,01,01)$, and $(01,00,00,00,10)$ were observed 98, 98, 99, 100, 100, and 99 times, respectively, out of 100 shots. In contrast, for the non-reachable target markings $(10,01,01,00,00)$, $(00,01,10,00,00)$, $(00,01,00,10,00)$, $(00,00,11,00,00)$, $(00,00,00,11,00)$, and $(00,00,00,00,11)$, no significant amplification of the specified target marking is observed. Instead, the measurement outcomes remain distributed among the reachable markings generated during Phase~I, and the specified non-reachable target does not appear as the dominant outcome. 

Although the measured counts differ slightly across the different non-reachable target cases, these variations are due to the finite number of measurement shots rather than target-specific amplitude amplification. These results are consistent with the theoretical behavior of Grover's search. Consequently, the proposed algorithm successfully distinguishes reachable from non-reachable markings while significantly increasing the probability of observing a reachable target marking. The $Q\#$ implementation used to obtain these results is provided in Appendix~\ref{app:qsharp}.

\begin{figure*}[t]
\centering

% First row
\begin{subfigure}[b]{0.32\textwidth}
\centering
\includegraphics[width=\linewidth]{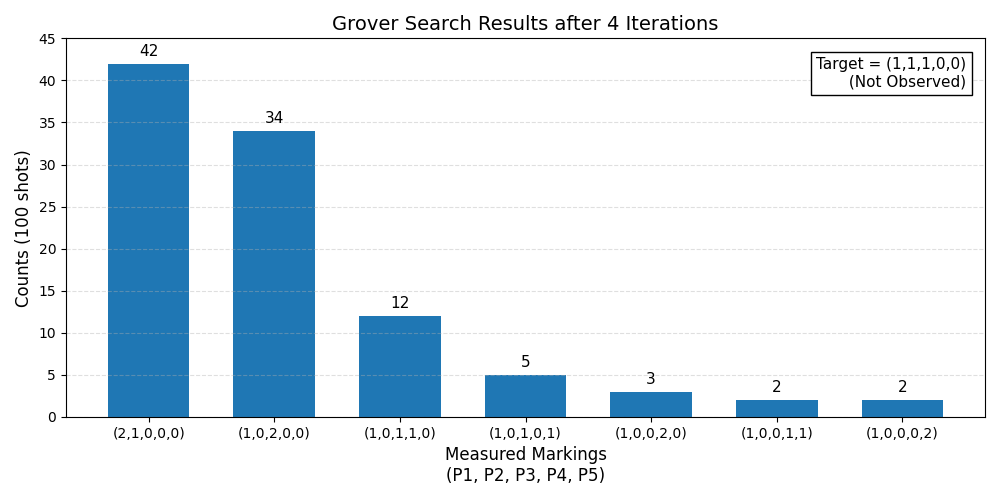}
\caption{Measurement outcomes for the non-reachable target marking
$(P_1,P_2,P_3,P_4,P_5)=(01,01,01,00,00)$.}
\end{subfigure}
\hfill
\begin{subfigure}[b]{0.32\textwidth}
\centering
\includegraphics[width=\linewidth]{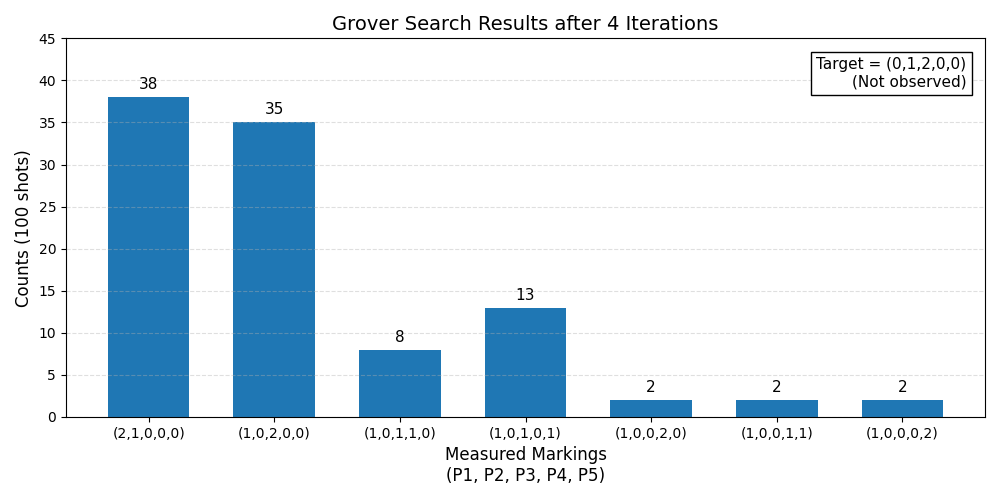}
\caption{Measurement outcomes for the non-reachable target marking
$(P_1,P_2,P_3,P_4,P_5)=(00,01,10,00,00)$.}
\end{subfigure}
\hfill
\begin{subfigure}[b]{0.32\textwidth}
\centering
\includegraphics[width=\linewidth]{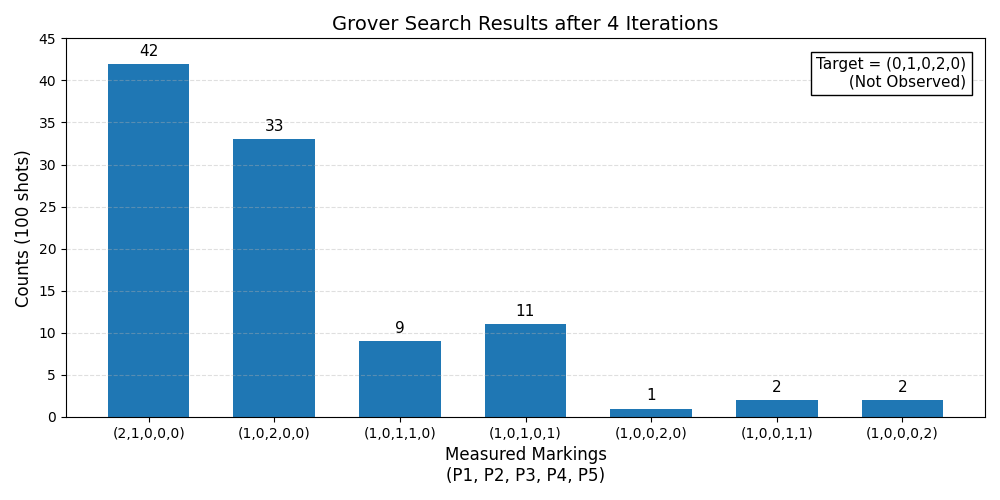}
\caption{Measurement outcomes for the non-reachable target marking
$(P_1,P_2,P_3,P_4,P_5)=(00,01,00,10,00)$.}
\end{subfigure}

\vspace{3mm}

% Second row
\begin{subfigure}[b]{0.32\textwidth}
\centering
\includegraphics[width=\linewidth]{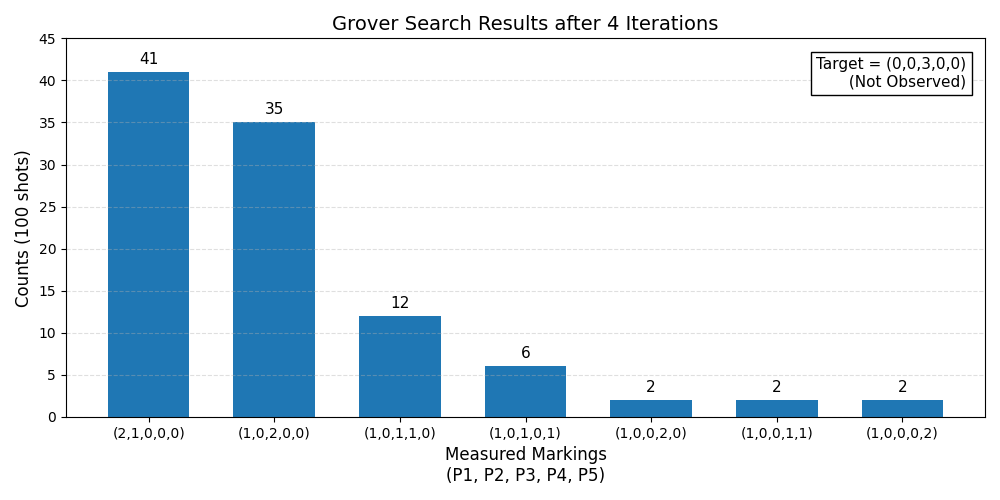}
\caption{Measurement outcomes for the non-reachable target marking
$(P_1,P_2,P_3,P_4,P_5)=(00,00,11,00,00)$.}
\end{subfigure}
\hfill
\begin{subfigure}[b]{0.32\textwidth}
\centering
\includegraphics[width=\linewidth]{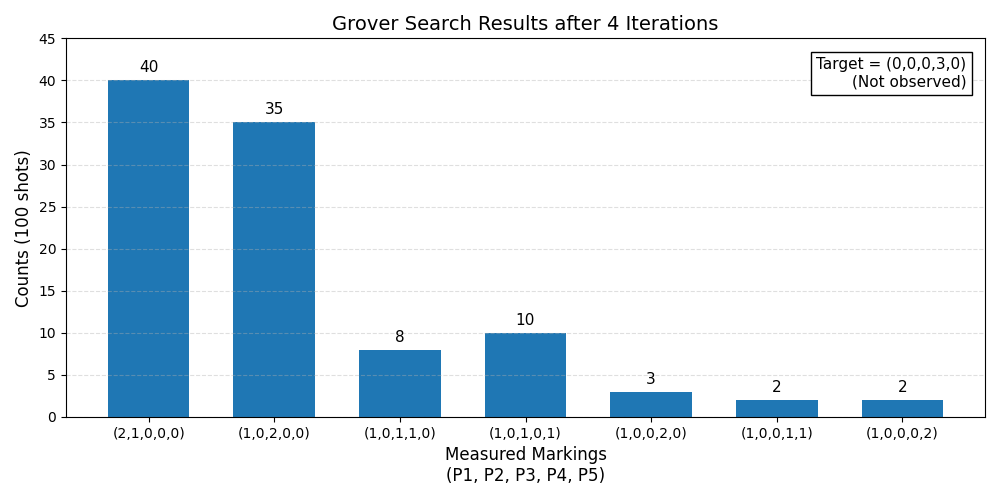}
\caption{Measurement outcomes for the non-reachable target marking
$(P_1,P_2,P_3,P_4,P_5)=(00,00,00,11,00)$.}
\end{subfigure}
\hfill
\begin{subfigure}[b]{0.32\textwidth}
\centering
\includegraphics[width=\linewidth]{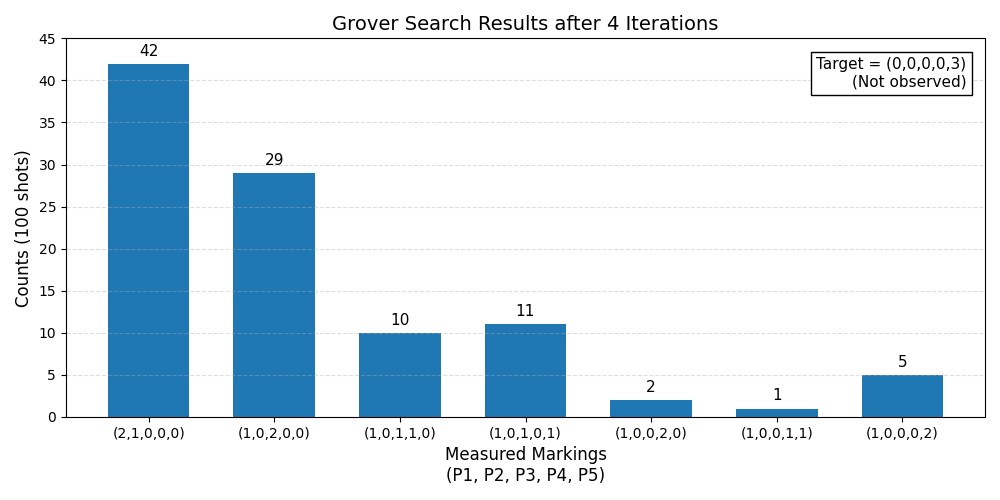}
\caption{Measurement outcomes for the non-reachable target marking
$(P_1,P_2,P_3,P_4,P_5)=(00,00,00,00,11)$.}
\end{subfigure}

\caption{Grover search results for non-reachable target markings.}
\label{fig:grover_utargets}
\end{figure*}

%\vspace{5pt}
\section{Complexity Analysis}

\vspace{-1pt}\noindent
Here we state the space and time complexities of Algorithm 1, assuming that it is executed on a quantum computer.

\vspace{5pt}
\noindent\textbf{Space Complexity:} 

Phase I: The marking register stores the  tokens in each of the
$m$ places. Since the QPN contains at most $n$ $q$-tokens, the number of $q$-tokens
at each place can be represented using
$\lceil\log_2(n+1)\rceil$ qubits. Therefore, the marking register requires $m\left\lceil\log_2(n+1)\right\rceil$
qubits, giving a space complexity of $O(m\log n).$ Since the QPN  model has $k$ transitions and we use $p$ ancillary transition registers, they require $O(p\log k)$ qubits. In addition, one firing qubit is associated with each iteration, requiring $O(p)$ qubits. Hence, the total space complexity during of Phase I is $O(m\log n+p\log k + p).$

Phase II: Grover's search operates on the marking register generated in Phase I. The marking register requires $(m\lceil \log_2(n+1)\rceil)$ qubits, giving a space requirement of $O(m\log n)$. This space has already been accounted for in the space complexity of Phase I. The oracle compares the current marking with the desired target marking. Depending on the implementation of the comparison and other quantum operations, additional ancillary qubits may be required. A straightforward implementation uses at most $O(m\log n)$ ancillary qubits. The diffusion operator acts on the same marking register and does not require asymptotically more space. Therefore, the space complexity of Phase II is $O(m\log n).$

Given that the registers used in Phase II do not asymptotically exceed those  required in Phase I, the overall space complexity of the proposed quantum reachability algorithm is $O(m\log n + p\log k + p).$

\noindent\textbf{Time Complexity:}

Phase~I consists of $p$ iterations. In each iteration, Hadamard gates are applied to the transition register, followed by the reversible transition operator $U_T$. Let $C_{U_T}$ denote the time it takes to apply $U_T$, including checking which transition is selected, determining whether it is enabled, and updating the corresponding marking. Thus, the time complexity of Phase~I is $O(p \times C_{U_T}).$ Since each place uses $\lceil\log_2(n+1)\rceil$ qubits to store its token count, checking or updating one place requires $O(\log n)$ operations. In the worst case, each of the $k$ transitions may involve all $m$ places. Therefore, one application of $U_T$ requires $O(k \times m \log n)$ operations. Hence, the worst-case time complexity of Phase~I is $O(p\times k \times m\log n).$

\indent
For Phase~II, the time complexity depends on the size $N$ of the computational search space generated during Phase~I.  The number of possible markings for $n$ q-tokens distributed among $m$ places is $\binom{n+m-1}{m-1}$. Consequently, an upper bound on the size of the search space is $N \leq \binom{n+m-1}{m-1}.$ The actual search space generated during Phase~I may be smaller than this upper bound because the marking and firing-register values are determined by the selected transitions and their enabling conditions. Therefore, not every possible combination of the marking, transition, and firing registers represents a valid execution of the QPN.

Grover amplitude amplification requires $O(\sqrt{N})$ iterations in the worst case. Each Grover iteration consists of one application of the oracle, followed by $A^\dagger$, the reflection operator $S_0$, and $A$. Therefore, the total time complexity of Phase II is
\[
O\!\left(
\sqrt{N}\left(C_O+2C_A+C_{S_0}\right)
\right),
\] 
where  (i) $C_A = O(p\times k \times m\log n)$ is the time complexity of Phase~I and $A^\dagger$ has the same asymptotic complexity as $A,$ (ii)   $ C_O = O(m\log n),$  and  (iii) $C_{S_0} = O(m\log n).$  Thus, the overall time complexity of Phase II is
\[
O\left(
\sqrt{N}
\left(
pkm\log n
\right)
\right).
\]

Including the initial state preparation performed in Phase~I, the overall time complexity of the proposed quantum algorithm is therefore
\[
O\left(
pkm\log n
+
\sqrt{N}
\left(
pkm\log n
\right)
\right) =  O(\sqrt{N}pkm\log n).
\]

\section{Concluding Remarks}\label{AA}

\noindent
This study has addressed the reachability problem in bounded Quantum Petri Nets by proposing a quantum algorithm that combines coherent state preparation with Grover's amplitude amplification. The proposed approach exploits quantum parallelism to construct a superposition of reachable markings from a given initial marking using $p$ iterations, and subsequently applies Grover's diffusion operator to amplify a desired target marking. Our theoretical analysis shows that the overall worst-case complexity of the proposed algorithm is upper bounded by $O(\sqrt{N}pkm\log n).$ Compared to the time complexity of $O(N k m)$ of the classical algorithm, where $N = \binom{n+m-1}{m-1}$ it is seen that our quantum algorithm gains  nearly a quadratic speed-up over it for large $n.$
A $Q\#$ implementation was developed to validate the proposed algorithm. The experimental results demonstrate that the state-preparation phase correctly generates a quantum superposition of reachable markings together with their corresponding transition-selection and firing histories, while the amplitude-amplification phase successfully amplifies the probability of measuring the desired target marking. These results confirm the correctness and feasibility of the proposed quantum algorithm for solving the bounded reachability problem in Quantum Petri Nets. It remains to be established if quantum computing models like the QPN may prove  beneficial in improving the time and space complexities of other computationally demanding problems.

\appendix

\section{Q\# Implementation of the Reachability Algorithm}
\label{app:qsharp}

This appendix presents the Q\# implementation used to validate the proposed reachability algorithm. The code is divided into three parts. Listing~\ref{lst:phase1} presents Phase I, where a superposition of reachable markings is generated. Listing~\ref{lst:phase2} presents Phase II, where amplitude amplification is applied to increase the probability of measuring the target place. Listing~\ref{lst:main} presents the execution and measurement procedure.

%--------------------------------------------------
% Phase II
%--------------------------------------------------
%--------------------------------------------------
% Phase I
%--------------------------------------------------
\begin{lstlisting}[
basicstyle=\scriptsize\ttfamily,
breaklines=true,
columns=fullflexible,
keepspaces=true,
frame=single,
caption={Phase I: Reachable-state preparation},
label={lst:phase1}
]
operation EncodeCount(place : Qubit[], count : Int) : Unit is Adj + Ctl {
    if ((count &&& 2) != 0) { X(place[0]); }
    if ((count &&& 1) != 0) { X(place[1]); }
}

operation ControlledIncrementWithTwoControls(c1 : Qubit, c2 : Qubit, place : Qubit[]) : Unit is Adj + Ctl {
    Controlled X([c1, c2, place[1]], place[0]);
    Controlled X([c1, c2], place[1]);
}

operation ControlledDecrementWithTwoControls(c1 : Qubit, c2 : Qubit, place : Qubit[]) : Unit is Adj + Ctl {
    Controlled X([c1, c2], place[1]);
    Controlled X([c1, c2, place[1]], place[0]);
}

operation ComputeNonZero(place : Qubit[], flag : Qubit) : Unit is Adj + Ctl {
    CNOT(place[0], flag);
    CNOT(place[1], flag);
    CCNOT(place[0], place[1], flag);
}

operation ComputeAlpha01(alpha : Qubit[], flag : Qubit) : Unit is Adj + Ctl {
    within { X(alpha[0]); }
    apply { CCNOT(alpha[0], alpha[1], flag); }
}

operation ComputeAlpha10(alpha : Qubit[], flag : Qubit) : Unit is Adj + Ctl {
    within { X(alpha[1]); }
    apply { CCNOT(alpha[0], alpha[1], flag); }
}

operation ComputeAlpha11(alpha : Qubit[], flag : Qubit) : Unit is Adj + Ctl {
    CCNOT(alpha[0], alpha[1], flag);
}

operation UT(marking : Qubit[], alpha : Qubit[], fired : Qubit) : Unit is Adj + Ctl {
    let p1 = marking[0..1];
    let p2 = marking[2..3];
    let p3 = marking[4..5];
    let p4 = marking[6..7];
    let p5 = marking[8..9];

    use nz = Qubit[3];
    use isT1 = Qubit();
    use isT2 = Qubit();
    use isT3 = Qubit();

    ComputeNonZero(p1, nz[0]);
    ComputeNonZero(p2, nz[1]);
    ComputeNonZero(p3, nz[2]);

    ComputeAlpha01(alpha, isT1);
    ComputeAlpha10(alpha, isT2);
    ComputeAlpha11(alpha, isT3);

    Controlled X([isT1, nz[0], nz[1]], fired);
    Controlled X([isT2, nz[2]], fired);
    Controlled X([isT3, nz[2]], fired);

    ComputeAlpha11(alpha, isT3);
    ComputeAlpha10(alpha, isT2);
    ComputeAlpha01(alpha, isT1);

    ComputeNonZero(p3, nz[2]);
    ComputeNonZero(p2, nz[1]);
    ComputeNonZero(p1, nz[0]);

    ComputeAlpha01(alpha, isT1);
    ComputeAlpha10(alpha, isT2);
    ComputeAlpha11(alpha, isT3);

    ControlledDecrementWithTwoControls(fired, isT1, p1);
    ControlledDecrementWithTwoControls(fired, isT1, p2);
    ControlledIncrementWithTwoControls(fired, isT1, p3);
    ControlledIncrementWithTwoControls(fired, isT1, p3);

    ControlledDecrementWithTwoControls(fired, isT2, p3);
    ControlledIncrementWithTwoControls(fired, isT2, p4);

    ControlledDecrementWithTwoControls(fired, isT3, p3);
    ControlledIncrementWithTwoControls(fired, isT3, p5);

    ComputeAlpha11(alpha, isT3);
    ComputeAlpha10(alpha, isT2);
    ComputeAlpha01(alpha, isT1);
}

operation PhaseI(
    marking : Qubit[],
    alpha1 : Qubit[],
    alpha2 : Qubit[],
    alpha3 : Qubit[],
    fired1 : Qubit,
    fired2 : Qubit,
    fired3 : Qubit
) : Unit is Adj + Ctl {
    let p1 = marking[0..1];
    let p2 = marking[2..3];
    let p3 = marking[4..5];
    let p4 = marking[6..7];
    let p5 = marking[8..9];

    EncodeCount(p1, 2);
    EncodeCount(p2, 1);
    EncodeCount(p3, 0);
    EncodeCount(p4, 0);
    EncodeCount(p5, 0);

    H(alpha1[0]);
    H(alpha1[1]);
    UT(marking, alpha1, fired1);

    H(alpha2[0]);
    H(alpha2[1]);
    UT(marking, alpha2, fired2);

    H(alpha3[0]);
    H(alpha3[1]);
    UT(marking, alpha3, fired3);
}
\end{lstlisting}
%--------------------------------------------------
% Phase II
%--------------------------------------------------
%--------------------------------------------------
% Phase II
%--------------------------------------------------
\begin{lstlisting}[
basicstyle=\scriptsize\ttfamily,
breaklines=true,
columns=fullflexible,
keepspaces=true,
frame=single,
caption={Phase II: Grover amplitude amplification},
label={lst:phase2}
]
operation OracleTargetMarking(marking : Qubit[], target : Int) : Unit is Adj + Ctl {
    let weights = [512, 256, 128, 64, 32, 16, 8, 4, 2, 1];

    within {
        for i in 0..9 {
            if ((target &&& weights[i]) == 0) {
                X(marking[i]);
            }
        }
    } apply {
        H(marking[9]);
        Controlled X(marking[0..8], marking[9]);
        H(marking[9]);
    }
}

operation ReflectionAboutZero(register : Qubit[]) : Unit is Adj + Ctl {
    let n = Length(register);

    within {
        for q in register {
            X(q);
        }
    } apply {
        H(register[n - 1]);
        Controlled X(register[0..n - 2], register[n - 1]);
        H(register[n - 1]);
    }
}

operation GroverIteration(
    marking : Qubit[],
    alpha1 : Qubit[],
    alpha2 : Qubit[],
    alpha3 : Qubit[],
    fired1 : Qubit,
    fired2 : Qubit,
    fired3 : Qubit,
    target : Int
) : Unit {
    let allRegs =
        marking
        + alpha1 + [fired1]
        + alpha2 + [fired2]
        + alpha3 + [fired3];

    OracleTargetMarking(marking, target);

    Adjoint PhaseI(
        marking,
        alpha1,
        alpha2,
        alpha3,
        fired1,
        fired2,
        fired3
    );

    ReflectionAboutZero(allRegs);

    PhaseI(
        marking,
        alpha1,
        alpha2,
        alpha3,
        fired1,
        fired2,
        fired3
    );
}
\end{lstlisting}
%--------------------------------------------------
% Execution and Measurement
%--------------------------------------------------
%--------------------------------------------------
% Main Program
%--------------------------------------------------
\begin{lstlisting}[
basicstyle=\scriptsize\ttfamily,
breaklines=true,
columns=fullflexible,
keepspaces=true,
frame=single,
caption={Main program: 100-shot execution},
label={lst:main}
]
operation MeasureMarkingAsInt(marking : Qubit[]) : Int {
    mutable value = 0;
    let weights = [512, 256, 128, 64, 32, 16, 8, 4, 2, 1];

    for i in 0..9 {
        if (M(marking[i]) == One) {
            set value += weights[i];
        }
    }

    return value;
}

@EntryPoint()
operation Main() : Unit {
    let target = 288;
    let groverIterations = 1;

    mutable count1001000000 = 0;
    mutable count0100100000 = 0;
    mutable count0100010100 = 0;
    mutable count0100010001 = 0;
    mutable count0100001000 = 0;
    mutable count0100000101 = 0;
    mutable count0100000010 = 0;
    mutable other = 0;

    for shot in 1..100 {
        use marking = Qubit[10];

        use alpha1 = Qubit[2];
        use alpha2 = Qubit[2];
        use alpha3 = Qubit[2];

        use fired1 = Qubit();
        use fired2 = Qubit();
        use fired3 = Qubit();

        let allRegs =
            marking
            + alpha1 + [fired1]
            + alpha2 + [fired2]
            + alpha3 + [fired3];

        PhaseI(marking, alpha1, alpha2, alpha3, fired1, fired2, fired3);

        for _ in 1..groverIterations {
            GroverIteration(
                marking,
                alpha1,
                alpha2,
                alpha3,
                fired1,
                fired2,
                fired3,
                target
            );
        }

        let measured = MeasureMarkingAsInt(marking);

        if (measured == 576) {
            set count1001000000 += 1;
        } elif (measured == 288) {
            set count0100100000 += 1;
        } elif (measured == 276) {
            set count0100010100 += 1;
        } elif (measured == 273) {
            set count0100010001 += 1;
        } elif (measured == 264) {
            set count0100001000 += 1;
        } elif (measured == 261) {
            set count0100000101 += 1;
        } elif (measured == 258) {
            set count0100000010 += 1;
        } else {
            set other += 1;
        }

        ResetAll(allRegs);
    }

    Message($"Grover iterations: {groverIterations}");
    Message($"Target decimal value: {target}");
    Message("Results after 100 shots:");
    Message($"1001000000 = (2,1,0,0,0), decimal 576: {count1001000000}");
    Message($"0100100000 = (1,0,2,0,0), decimal 288: {count0100100000}");
    Message($"0100010100 = (1,0,1,1,0), decimal 276: {count0100010100}");
    Message($"0100010001 = (1,0,1,0,1), decimal 273: {count0100010001}");
    Message($"0100001000 = (1,0,0,2,0), decimal 264: {count0100001000}");
    Message($"0100000101 = (1,0,0,1,1), decimal 261: {count0100000101}");
    Message($"0100000010 = (1,0,0,0,2), decimal 258: {count0100000010}");
    Message($"Other: {other}");
}
\end{lstlisting}

\end{document}